\documentclass[11pt]{article}
\usepackage{amsmath,amsthm,amssymb}
\newcommand{\remove}[1]{}
\newtheorem{thm}{Theorem}[section]
\newtheorem{claim}[thm]{Claim}
\newtheorem{lem}[thm]{Lemma}

\newtheorem{THM}{Theorem}

\def\D{{\partial}}

\newcommand{\Z}{{\mathbb Z}}

\newcommand{\E}{{\mathbb E}}
\newcommand{\Cplx}{{\mathbb C}}
\newcommand{\poly}{{\mathrm{poly}}}

\newcommand{\eps}{\varepsilon}

\newcommand{\LCM}{\mathrm{LCM}}

\newcommand{\F}{\mathbb{F}}

\begin{document}

\title{Affine extractors over large fields with exponential error}
\author{ Jean Bourgain \thanks{School of Mathematics, Institute for Advanced Study,
Princeton, NJ. Email: \texttt{bourgain@ias.edu}. Research supported by NSF grant DMS 1301619.}\and Zeev Dvir \thanks{Department of Computer Science and Department of Mathematics, Princeton University, Princeton NJ.
Email: \texttt{zeev.dvir@gmail.com}. Research partially
supported by NSF grants CCF-0832797, CCF-1217416 and by the Sloan fellowship.} \and Ethan Leeman \thanks{Department of Mathematics, University of Texas at Austin. Email: \texttt{eleeman@math.utexas.edu}}}
\date{}
\maketitle

\begin{abstract}
We describe a construction of explicit affine extractors over large finite fields with exponentially small error and linear output length. Our construction relies on a deep theorem of Deligne giving tight  estimates for exponential sums over smooth varieties in high dimensions. 
\end{abstract}

\section{Introduction}

An {\em affine extractor} is a mapping $E : \F_q^n \mapsto \{0,1\}^m$, with $\F_q$ the field of $q$ elements, such that for any subspace $V \subset \F_q^n$ of some fixed dimension $k$, the output of $E$ on a uniform sample from $V$ is distributed close to uniformly over the image. More precisely, if $X_V$ is a random variable distributed uniformly on $V$, then $E(X_V)$ is $\eps$-close, in statistical distance\footnote{The {\sf statistical distance} between two distributions $P$ and $Q$ on a finite domain $\Omega$ is defined as
$\mathop{\max}_{S\subseteq \Omega} \left|P(S) - Q(S)\right|
.$ We say that $P$ is $\eps$-{\sf close} to $Q$ if the
statistical distance between $P$ and $Q$ is at most $\eps$.}, to the uniform distribution over $\{0,1\}^m$ (here, and in the following, we will often identify a random variable with its distribution). It is easy to show that a random function $E$ will be an affine extractor. However, constructing {\em explicit} families of affine extractors is a challenging problem which is still open for many settings of the parameters. By explicit, we mean that the mapping $E$ can be computed deterministically and efficiently, given the parameters $n,k$ and $q$. 

The task of constructing explicit affine extractor is an instance of a more general set of problems in which one has a combinatorial or algebraic object possessing certain `nice' properties, one would expect to have in a random (or generic) object, and wishes to come up with an explicit instance of such an object. Other examples include expander graphs \cite{RVW02, LPS88}, Ramsey graphs \cite{BRSW06}, Error correcting codes, and other variants of algebraic extractors (e.g., extractors for polynomial sources \cite{DGW09,BG12} or varieties \cite{Dvir08}). Explicit constructions of these `pseudo-random' objects have found many (often surprising) applications in theoretical computer science and mathematics (see, e.g.,  \cite{HLW06} for some examples).

Ideally we would like to be able to give explicit constructions of affine extractors for any given $n,k,q$ with  output length $m$ as large as possible and with error parameter $\eps$ as small as possible. It is not hard to show, using the probabilistic method, that there {\em exist} affine extractors with $m$ close to $ k \cdot \log(q)$ and $\eps = q^{\Omega(-k)}$ over any finite field and for $k$ as small as $O(\log(n))$. Matching these parameters with an explicit construction is still largely open. 

When the size of the field is fixed ($q$ is a constant and $n$ tends to infinity) a construction of Bourgain \cite{Bour07} (see also \cite{Yeh11,Li11}) gives affine extractors with $m = \Omega(k)$ and $\eps = q^{\Omega(-k)}$ whenever $k \geq \Omega(n)$ ($k$ can actually be slightly sub linear in $n$). For smaller values of $k$, there are no explicit constructions of extractors (even with $m=1$) over small fields (see Theorem C in \cite{Bour10}  for a related result handling intermediate field sizes).  When the size of the field $\F_q$ is allowed to grow with $n$ more is known. Gabizon and Raz \cite{gabizon2011deterministic} were the first to consider this case and showed an explicit constructions when $q > n^c$, for some constant $c$. Their construction achieves nearly optimal output length but with error $\eps = q^{-\Omega(1)}$ instead of $q^{\Omega(-k)}$. 
   
  
The purpose of this note is to give a construction of an explicit affine extractor  for $q > n^{C \cdot \log \log n}$ with error $q^{\Omega(-k)}$ and output length $m$ close to $(1/2)k \log(q)$ bits. It will be more natural to consider the extractor as a mapping $E : \F_q^n \mapsto \F_q^m$ instead of with image $\{0,1\}^m$ and so we will aim to have output length $m$ close to $k/2$ (since each coordinate of the output is composed of roughly $\log(q)$ bits).

The construction does not work for any finite field $\F_q$. Firstly, we will only consider prime $q$. We will also need the property that $q-1$ does not have too many prime factors. We expect due to a result by Prachar \cite{halberstam1956distribution} that $q-1$ will have approximately $\log \log q$ distinct prime divisors: Prachar, in Halberstam's paper, proved that if $\omega(q-1)$ is the number of distinct prime factors of $q-1$, then $$\sum_{q \leq n} \omega(q-1) = \left(1 + o(1) \right) \frac{n}{\log n} \log \log n.$$ Therefore, the average number of distinct prime divisors of $q-1$ for most $q$ is $O(\log \log q)$, but some primes may have as many as $\frac{\log q}{\log \log q}$ distinct prime factors. We say that a prime $q$ is $\bf{typical}$ if $q-1$ has $O(\log \log q)$ distinct prime factors\footnote{The constant in the big `O' can be arbitrary at the cost of increasing the constant $C$ in Theorem~\ref{THM-main}.}.
 
\begin{THM}\label{THM-main}
For any $\beta \in (0,1/2)$ there exists $C>0$ so that the following holds:	Let $k \leq n$ be integers  and let $q$ be a typical prime such that $q > n^{C\log \log n}$. Then, if $m = \lfloor \beta k \rfloor$, there is an explicit function $E: \F_q^n \to \F_q^{m}$ such that for any $k$-dimensional affine subspace $V$ in $\F_q^n,$ if $X_V$ is a uniform random variable on $V$, then  $E(X_V)$ is $q^{- \Omega(k)}$-close to the uniform distribution.
\end{THM}

The rest of the paper is organized as follows: In Section~\ref{sec-construction} we describe the construction of the extractor. In Section~\ref{sec-analysis} we prove that the output of the extractor is close to uniform, whenever the ingredients of the construction satisfy certain conditions. In Section~\ref{sec-explicit} we discuss the explicitness of the construction  and in Section~\ref{sec-final} we combine all of these results to prove Theorem~\ref{THM-main}.

\section{The construction}\label{sec-construction}

The construction will be given by a polynomial mapping $F_{d,A}: \F_q^n \mapsto \F_q^m$. This mapping will take as parameters two objects. The first is a list of positive integers $d = (d_1,\ldots,d_n)$ and the other is an $m \times n$ matrix $A = (a_{ij})$. The mapping is then defined as

\begin{eqnarray*}
F_{d,A}(x_1,\ldots,x_n) &=&
\begin{pmatrix}
a_{11} & \cdots & a_{1n} \\
\vdots & \ddots & \vdots \\
a_{m1} & \cdots & a_{mn} 
\end{pmatrix}
\begin{pmatrix}
x_1^{d_1} \\ x_2^{d_2} \\ \vdots \\ x_n^{d_n}
\end{pmatrix} \\
&=& \left(\sum_{j=1}^n a_{1j}x_j^{d_j}, \ldots,   \sum_{j=1}^n a_{mj}x_j^{d_j} \right)^t
\end{eqnarray*}
 
This can be also written as $F_{d,A}(x) = A \cdot x^d$, where we interpret $x^d$ as being coordinate-wise exponentiation.

We will show below that, if $d$ and $A$ satisfy certain  conditions, the output $F_{d,A}(X_V)$ is exponentially close to uniform, whenever $X_V$ is uniformly distributed over a $k$ dimensional subspace.

\section{The analysis}\label{sec-analysis}
 
In this section we prove  that the function $F_{d,A}(x)$ defined above is indeed an affine extractor for carefully chosen $d$ and $A$. In the next section we will discuss the complexity of finding such $d$ and $A$ efficiently.

\begin{thm}\label{thm-analysisFdA}
For every $\beta < 1/2$ there exists $\eps > 0$ such that the following holds: Let $q$ be prime and let $m \leq k \leq n$ be integers with $m = \lfloor \beta k \rfloor$. Let $A$  be an $m \times n$ matrix over $\F_q$ in which every $m$ columns are linearly independent. Let $d = (d_1,\ldots,d_n) \in \Z_{>0}^n$ be such that $ \LCM(d_1,\ldots,d_n) \leq  q^{\eps}$ and such that $d_1,\ldots,d_n$ are all distinct and co-prime to $q-1$. Then, for any $k$-dimensional  affine subspace $V \subset \F_q^n$, if $X_V$ is uniformly distributed over $V$ then $F_{d,A}(X_V)$ is $q^{-(\eps/2) k}$-close to uniform.  
\end{thm}

\subsection{Preliminaries} 

We start by setting notations and basic properties of the discrete Fourier transform over $\F_q^m$. For $c = (c_1,\ldots,c_m) \in \F_q^m$ we define the additive character
$ \chi_c(x) : \F_q^m \mapsto \Cplx^*$ as $\chi_c(x) =  \omega_q^{c \cdot x}$ where $c \cdot x = \sum_{i=1}^m c_i x_i$ and $\omega_q = e^{2\pi i/q}$ is a primitive root of unity of order $q$.

The following folklore result (known in the extractor literature as a XOR lemma) gives sufficient conditions for a distribution to be close to uniform. The simple proof can be found in \cite{rao2007exposition} for example.

\begin{lem}\label{lem-XOR} Let $X$ be a random variable distributed  over $\F_q^m$ and suppose  that $\left| \E \left[ \chi_c(X) \right] \right| \leq \eps$ for every non-zero $c \in \F_q^m$. Then $X$ is $\eps \cdot q^{m/2}$ close, in statistical distance, to the uniform distribution over $\F_q^m$.
\end{lem}

The next powerful theorem is a special case of a theorem of  Deligne \cite{Deligne74} (see  \cite{MK93} for a 
statement of the theorem in the form we use here). Before stating
the theorem we will need the following definition: Let $f \in
\F_q[x_1,\ldots,x_n]$ be a homogenous polynomial. We say that $f$ is
{\sf smooth} if the only common zero of the (homogenous) $n$
partial derivatives $\frac{\D f}{\D x_i}(x)$, $\, i \in [n]$ over the algebraic closure of $\F_q$, is
the all zero vector. 

\begin{thm}[Deligne]\label{thm-deligne}
Let $f \in \F_q[x_1,\ldots,x_n]$ be a polynomial of degree $d$ and
let $f_d$ denote its homogenous part of degree $d$. Suppose $f_d$
is smooth. Then, for every non-zero $b \in \F_q$ we have
\[ \left| \sum_{x \in \F_q^n} \chi_b(f(x)) \right| \leq (d-1)^n \cdot
q^{n/2}. \]
\end{thm}

Another simple lemma we will use in the proof shows how to parameterize a given subspace $V \subset \F_q^n$ in a convenient way as the image of a particular linear mapping.

\begin{lem}\label{lem-basis} Let $V \subset \F_q^n$ be a $k$-dimensional affine subspace. Then, there exists an affine map $\ell = (\ell_1,\ldots,\ell_n) : \F_q^k \to \F_q^n$ whose image is $V$ such that the following holds: There exists $k$ indices $1 \leq j_1 < j_2 < \ldots < j_k \leq n$ such that

\begin{enumerate}
\item For all $i \in [k], \ell_{j_i}(t) = t_i.$
\item If $j < j_1$, then $\ell_j(t) \in \F_q.$
\item If $j < j_i$ for $i > 1$ then $\ell_j(t)$ is an affine function just of the variables $t_1, t_2, \ldots, t_{i-1}.$
\end{enumerate}

\end{lem}

\begin{proof}
The mapping $\ell$ can be defined  greedily as follows. Let $j_1$ be the smallest index so that the $j_1$'th coordinate of $V$ is not constant. We let $\ell_{j_1}(t) = t_1$ and continue to find the next smallest coordinate so that the $j_2$'th coordinate of $V$ is not a function of $t_1$. Set $\ell_{j_2}(t)= t_2$ and continue in this fashion to define the rest of the mapping.
\end{proof}

\subsection{Proof of Theorem~\ref{thm-analysisFdA}}

Let $Z = F_{d,A}(X_V)$ denote the random variable over $\F_q^m$ obtained by applying $F_{d,A}$ on a uniform sample from the subspace $V$. Observe that, w.l.o.g., we can assume $$ d_1 > d_2 > \ldots > d_n$$ since permuting the columns of $A$ keeps the property that every $m$ columns are linearly independent.

Let $\ell : \F_q^k \mapsto \F_q^n$ be an affine mapping satisfying the conditions of Lemma~\ref{lem-basis} so that the image of $\ell$ is $V$. Thus, there is a set $S \subset [n]$ of size $|S|=k$ so that, if $S = \{j_1 < \ldots < j_k\}$, the coordinates of the mapping $\ell$ satisfy the three items in the lemma.

Let $c = (c_1,\ldots,c_m) \in \F_q^m$ be a non zero vector. We will proceed to give a bound on the expectation $\left| \E[\chi_c(Z) ] \right|$ and then use Lemma~\ref{lem-XOR} to finish the proof. To that end, let $b = (b_1,\ldots,b_n)$ be given by the product $c^t \cdot A$ (multiplying $A$ from the left by the transpose of $c$). Then, 
$$ \chi_c(F_{d,A}(x)) = \chi_1( b \cdot x^d) = \omega_q^{b_1x_1^{d_1} + \ldots + b_n x_n^{d_n}}. $$

Therefore, 
\begin{equation}\label{eq-expsum1}
\left| \E[\chi_c(Z) ] \right| = \left| q^{-k} \sum_{t_1,\ldots,t_k \in \F_q} \chi_1\left( b_1\ell_1(t)^{d_1} + \ldots + b_n \ell_n(t)^{d_n}   \right) \right|.
\end{equation}

We will now perform an invertible (non-linear) change of variables on the above exponential sum to bring it to a more convenient form. Let 
$$ D  = \LCM( d_{j_1}, \ldots,d_{j_k}). $$ and let $D_i = D/d_{j_i}$ for $i=1\ldots k$. The change of variables is given by 
$$ s_i^{D_i} = t_i, \,\,\, i \in [k]. $$ Observe that this is an invertible change of variables since the $d_i$'s are all co-prime to $q-1$ (and hence the numbers $D_i$ are as well). Specifically, we have $s_i = t_i^{D_i^{-1} \mod q-1}$.

 Let us denote by
$$ \tilde \ell_j(s) = \ell_j(s_1^{D_1},\ldots,s_k^{D_k}). $$
Changing variables in (\ref{eq-expsum1}) now gives
\begin{equation}\label{eq-expsum2}
\left| \E[\chi_c(Z) ] \right| = \left| q^{-k} \sum_{s_1,\ldots,s_k \in \F_q} \chi_1\left( b_1\tilde\ell_1(s)^{d_1} + \ldots + b_n \tilde\ell_n(s)^{d_n}   \right) \right|.	
\end{equation}

\begin{claim}\label{cla-changevars}
The functions $\tilde \ell_i^{d_i}(s)$, $i \in [n]$, satisfy the following: 
\begin{enumerate}
	\item For all $i \in [k]$ we have $\tilde \ell_{j_i}^{d_{j_i}}(s) = s_i^D$.
	\item For all $j \not\in S$ the function $\tilde \ell_j^{d_j}(s)$ is a polynomial in $s_1,\ldots,s_k$ of total degree less than $D$.
\end{enumerate}
\end{claim}
\begin{proof}
To see the first item, let $i \in [k]$ and observe that $\ell_{j_i}(t) = t_i$. Thus, $$\tilde \ell_{j_i}^{d_{j_i}}(s) = \left(s_i^{D/d_{j_i}}\right)^{d_{j_i}} = s_i^D.$$
For the second item, let $j \not\in S$ and suppose $j_i < j < j_{i+1}$ for some $i \in [k]$ (a similar argument will work for the two cases $j < j_1$ and $j > j_k$). By Lemma~\ref{lem-basis}, the affine function $\ell_j(t)$ depends only on the variables $t_1,\ldots,t_i$. Thus, the maximum degree obtained in $\tilde \ell_j^{d_j}(s)$ is bounded by $$d_j \cdot \max\{D_1,\ldots,D_i\} = d_j \cdot D_i = D \cdot (d_j/d_{j_i}) < D. $$ 
\end{proof}

In view of the last claim, we can write (\ref{eq-expsum2}) as
\begin{equation}\label{eq-expsum3}
\left| \E[\chi_c(Z) ] \right| = \left| q^{-k} \sum_{s_1,\ldots,s_k \in \F_q} \chi_1\left( b_{j_1} s_1^D + \ldots + b_{j_k} s_k^D + g(s)   \right) \right|,	
\end{equation}
where $g(s)$ is a polynomial of total degree less than $D$. If we knew that all of $b_{i_1}, \ldots,b_{i_k}$ were non zero we could have applied Deligne's result (Theorem~\ref{thm-deligne}) and complete the proof (since the polynomial in the sum is clearly smooth). However, since $b = c^t \cdot A$ for an arbitrary non-zero $c \in \F_q^m$, $b$ might have some coordinates equal to zero. However, since every $m$ columns of $A$ are linearly independent, we have that the vector $b = (b_1,\ldots,b_n)$ can have at most $m-1 < k/2$ zero coordinates (otherwise $c$ would be orthogonal to at least $m$ columns). Hence, out of the $k$ values $b_{i_1}, \ldots,b_{i_k}$, at least $k/2$ are non zero. Suppose w.l.o.g that these are the first $k/2$ (if $k$ is odd we need to add the floor function below for $k/2$). We can now break the sum in (\ref{eq-expsum3}) using the triangle inequality as follows
\begin{eqnarray}\label{eq-expsum4}
\left| \E[\chi_c(Z) ] \right| = q^{-k/2} \sum_{s_{k/2+1},\ldots,s_k \in \F_q}\left| q^{-k/2} \sum_{s_1,\ldots,s_{k/2} \in \F_q} \chi_1\left( \sum_{i \in [k/2]} b_{j_i} s_i^D  + g_{s_{k/2+1},\ldots,s_k}(s_1,\ldots,s_{k/2})   \right) \right|,	
\end{eqnarray}
with $g_{s_{k/2+1},\ldots,s_k}(s_1,\ldots,s_{k/2})$ a polynomial in $s_1,\ldots,s_{k/2}$ of degree less than $D$. In each of the inner sums we have a smooth polynomial of degree $D$ in the ring $\F_q[s_1,\ldots,s_{k/2}]$ and so, applying Theorem~\ref{thm-deligne} on each of them (and recalling that $D \leq q^\eps$), we obtain 
\begin{equation}\label{eq-expsum5}
\left| \E[\chi_c(Z) ] \right| \leq q^{-k/2} \cdot (D-1)^{k/2} \cdot q^{k/4} \leq q^{(-1/4 + \eps/2)k}
\end{equation}

Using Lemma~\ref{lem-XOR}, and setting $\eps = 1/4 - \beta/2 > 0$, we now get that $Z$ has statistical distance at most $$ q^{(-1/4 + \eps/2)k} \cdot q^{m/2} \leq q^{(-1/4 + \eps/2 + \beta/2)k} \leq q^{-(\eps/2)k} $$ from the uniform distribution on $\F_q^m$. This completes the proof of Theorem~\ref{thm-analysisFdA}. \qed

\section{Explicitness of $F_{d,A}$}\label{sec-explicit}

The explicitness of the construction requires us to give a deterministic, efficient, algorithm to produce a matrix $A$ and a sequence of integers $d_1,\ldots,d_n$ satisfying the conditions of Theorem~\ref{thm-analysisFdA}. 

Finding an $m \times n$ matrix in which each $m \times m$ sub matrix is invertible can be done efficiently as long as $q$, the field size, is sufficiently large. For example, one can take a Vandermonde matrix with $a_{ij} = r_j^{i-1}$ for any set of distinct field elements $r_1,\ldots,r_n \in \F_q$. 

To find a sequence $d = (d_1,\ldots,d_n)$ we will have to make some stronger assumption about $q$. This is summarized in the following lemma.
  
\begin{lem}\label{lem-d}
For any $\eps > 0$ there exists $C >0$ such that the following holds: There is a deterministic algorithm that, given integer inputs $n,q,k$ where $k < n < q,$ $q$ a typical prime such that $q > n^{C \log \log n},$ runs in $\poly(n)$ time and returns $n$ integers $d_1 > \ldots > d_n > 1$ all co-prime to $q-1$  with $\LCM(d_1,\ldots,d_n) < q^\eps$.
\end{lem}
\begin{proof}  
Let $D$ be the product of the first $\lceil \log_2(n+1) \rceil$ primes that are co-prime with $q-1.$ Let $d_1 > \ldots > d_n$ be $n$ distinct divisors of $D.$ If $q-1$ has at most $C'\log\log(q)$ prime factors, $D$ can be upper bounded by the product of the first $\log n + C'(\log \log q)$ primes . By the Prime Number Theorem, $$D < \left(n \left(\log q \right)^{C'} \right) ^{C'' \log \log \left( n \left(\log q \right)^{C'} \right)}$$ for some constant $C'' > 0.$ 
Now for any $\eps,C'', C'$, we can pick a sufficiently large $C$ such that, if $q > n^{C\log\log n}$ this  expression is at most $q^{\eps}$.
\end{proof}

\section{Proof of Theorem~\ref{THM-main}}\label{sec-final}

We now put all the ingredients together to prove Theorem~\ref{THM-main}.  Given $m = \lfloor \beta k \rfloor$ we let $\eps = 1/4 - \beta/2$ and, using Lemma~\ref{lem-d} find a sequence of integers $d_1,\ldots,d_n$ all coprime to $q-1$ so that their product is at most $q^\eps$. We let $A$ be an $m \times n$ Vandermonde matrix and define $E(x) = F_{d,A}(x)$. Using Theorem~\ref{thm-analysisFdA} we get that $E(X_V)$ is $q^{-(\eps/2)k}$-close to the uniform distribution on $\F_q^m$.


\bibliographystyle{alpha}

\bibliography{delignext}

\begin{thebibliography}{BRSW06}

\bibitem[Bou07]{Bour07}
J.~Bourgain.
\newblock On the construction of affine extractors.
\newblock {\em Geometric And Functional Analysis}, 17(1):33--57, 2007.

\bibitem[Bou10]{Bour10}
J.~Bourgain.
\newblock On exponential sums in finite fields.
\newblock In I.~B\'ar\'any, J.~Solymosi, and G.~S\'agi, editors, {\em An
  Irregular Mind}, volume~21 of {\em Bolyai Society Mathematical Studies},
  pages 219--242. Springer Berlin Heidelberg, 2010.

\bibitem[BRSW06]{BRSW06}
B.~Barak, A.~Rao, R.~Shaltiel, and A.~Wigderson.
\newblock 2-source dispersers for sub-polynomial entropy and ramsey graphs
  beating the frankl-wilson construction.
\newblock In {\em Proceedings of the thirty-eighth annual ACM symposium on
  Theory of computing}, pages 671--680, New York, NY, USA, 2006. ACM Press.

\bibitem[BSG12]{BG12}
E.~Ben-Sasson and A.~Gabizon.
\newblock Extractors for polynomials sources over constant-size fields of small
  characteristic.
\newblock In {\em APPROX-RANDOM}, Lecture Notes in Computer Science, pages
  399--410. Springer, 2012.

\bibitem[Del74]{Deligne74}
P.~Deligne.
\newblock La conjecture de {W}eil.
\newblock {\em I , Inst. Hautes ´Etudes Sci. Publ. Math.}, 43:273--307, 1974.

\bibitem[DGW09]{DGW09}
Z.~Dvir, A.~Gabizon, and A.~Wigderson.
\newblock {Extractors And Rank Extractors For Polynomial Sources}.
\newblock {\em Comput. Complex.}, 18(1):1--58, 2009.

\bibitem[Dvi12]{Dvir08}
Z.~Dvir.
\newblock {Extractors for varieties}.
\newblock {\em Comput. Complex.}, 21:515--572, 2012.

\bibitem[GR08]{gabizon2011deterministic}
A.~Gabizon and R.~Raz.
\newblock Deterministic extractors for affine sources over large fields.
\newblock {\em Combinatorica}, 28(4):415--440, 2008.

\bibitem[Hal56]{halberstam1956distribution}
H.~Halberstam.
\newblock {On the Distribution of Additive Number-Theoretic Functions (II)}.
\newblock {\em {Journal of the London Mathematical Society}}, 1(1):1--14, 1956.

\bibitem[HLW06]{HLW06}
S.~Hoory, N.~Linial, and A.~Wigderson.
\newblock Expander graphs and their applications.
\newblock {\em Bull. Amer. Math. Soc.}, 43:439--561, 2006.

\bibitem[Li11]{Li11}
X.~Li.
\newblock A new approach to affine extractors and dispersers.
\newblock In {\em IEEE Conference on Computational Complexity}, pages 137--147.
  IEEE Computer Society, 2011.

\bibitem[LPS88]{LPS88}
A.~Lubotzky, R.~Phillips, and P.~Sarnak.
\newblock Ramanujan graphs.
\newblock {\em Combinatorica}, 8(3):261--277, 1988.

\bibitem[MK93]{MK93}
O.~Moreno and P.~Kumar.
\newblock Minimum distance bounds for cyclic codes and {D}eligne's theorem.
\newblock {\em IEEE Transactions on Information Theory}, 39(5):1524--1534,
  1993.

\bibitem[Rao07]{rao2007exposition}
A.~Rao.
\newblock {An Exposition of BourgainÕs 2-Source Extractor}.
\newblock In {\em {Electronic Colloquium on Computational Complexity (ECCC)}},
  volume~14, page 034, 2007.

\bibitem[RVW02]{RVW02}
O.~Reingold, S.~Vadhan, and A.~Wigderson.
\newblock Entropy waves, the zig-zag graph product, and new constant-degree
  expanders.
\newblock 155(1):157--187, 2002.

\bibitem[Yeh11]{Yeh11}
A.~Yehudayoff.
\newblock Affine extractors over prime fields.
\newblock {\em Combinatorica}, 31(2):245--256, 2011.

\end{thebibliography}

\end{document}